\documentclass[runningheads]{llncs}
\usepackage{latexsym}
\input epsf

\begin{document}

\title{Monomial Testing and Applications}

\titlerunning{Monomial Testing and Applications}

\author{Shenshi Chen}

\authorrunning{S. Chen}

\institute{Department of Computer Science,
University of Texas-Pan American, Edinburg, TX 78539, USA.
\email{schen@broncs.utpa.edu}}

\maketitle

\begin{abstract}
 In this paper, we devise two algorithms for the problem of testing $q$-monomials
 of degree $k$ in any multivariate
polynomial represented by a circuit, regardless of the primality of $q$. One is
an $O^*(2^k)$ time randomized algorithm.
The other is an $O^*(12.8^k)$ time deterministic algorithm
for the same $q$-monomial testing problem but requiring the polynomials to be represented by tree-like circuits.
Several applications of $q$-monomial testing are also given, including a deterministic
$O^*(12.8^{mk})$ upper bound for the $m$-set $k$-packing problem.
\\
\\
{\bf Keywords:} Group algebra; complexity; multivariate polynomials; monomials; monomial testing;
 randomized algorithms; derandomization.
\end{abstract}

\section{Introduction}

Recent research on testing multilinear
monomials and $q$-monomials in multivariate polynomials \cite{koutis08,williams09,chen11,chen11a,chen11b,chen12a,chen12b} requires
that $Z_q$ be a field, which is true when $q\ge 2 $ is prime.
When $q>2$ is not prime, $Z_q$ is no longer a field, hence
the group algebra based approaches in \cite{koutis08,williams09,chen11b,chen12b}
become inapplicable.
When $q$ is not prime, it remains open whether the problem of testing $q$-monomials
in a multivariate polynomial can be solved in some compatible complexity,
such as $O^*(c^k)$ time for a constant $c\ge 2$.
Our work in \cite{bill12} presents
a randomized $O^*(7.15^k)$ algorithm for testing $q$-monomials of degree $k$
in a multivariate polynomial that is represented by a tree-like circuit.
This algorithm works for any fixed integer $q\ge 2$, regardless of $q$'s primality.
Moreover, for prime $q>7$,
it provides us with some substantial improvement on the time complexity of
the previously known algorithm \cite{chen11b,chen12b} for testing $q$-monomials.

Randomized algebraic techniques have recently led to the once fastest
randomized algorithms of time $O^*(2^k)$ for the  $k$-path
problem and other problems \cite{koutis08,williams09}. Another recent
seminal example is the improved $O(1.657^n)$ time randomized algorithm for the Hamiltonian
path problem by Bj\"{o}rklund~\cite{Bjorklund2010}. This algorithm provided a positive answer to
the question of whether the Hamiltonian path problem can be solved in time $O(c^n)$ for some constant
$1<c < 2$, a challenging problem that had been open for
half of a century. Bj\"{o}rklund {\em et al.} further extended the above randomized algorithm
to the $k$-path testing problem with $O^*(1.657^k)$ time complexity \cite{Bjorklund2010b}.
Very recently, those two algorithms were simplified further by Abasi and Bshouty \cite{bshouty13}.

This paper consists of three key contributions: The first is an $O^*(2^k)$ time randomized algorithm that
gives an affirmative answer to the $q$-monomial testing problem for polynomials represented by circuits,
regardless of the primality of $q\ge 2$. We generalize the circuit reconstruction and
variable replacements proposed in \cite{bill12} to transform the $q$-monomial testing problem,
for polynomials represented by a circuit,
into the multilinear monomial testing problem and furthermore enabling
the usage of the group algebraic approach originated
by Koutis \cite{koutis08} to help resolve the $q$-monomial testing problem.
The second is an  $O^*(12.8^k)$ deterministic algorithm for testing $q$-monomials
in multivariate polynomials represented by tree-like circuits.
Inspired by the work in \cite{chen11b,chen12b}, we devise this deterministic algorithm by derandomizing
the first randomized algorithm for tree-like circuits with the help of
the perfect hashing functions by Chen {\em et al.}
\cite{jianer-chen07} and the deterministic polynomial identity testing algorithm by Raz
and Shpilka \cite{raz05} for noncommunicative polynomials.
The third is to exhibit several applications of
$q$-monomial testing to designing algorithms for concrete problems.
Specifically, we show how $q$-monomial testing can be applied to
the non-simple $k$-path testing problem, the generalized $m$-set $k$-packing problem,
and the generalized $P_2$-Packing problem. In particular, we design
a deterministic algorithm for solving the $m$-set $k$-packing problem in $O^*(12.8^{mk})$,
which is, to our best knowledge, the best upper bound for deterministic algorithms
to solve this problem.

\section{Notations and Definitions}
For variables $x_1, \dots, x_n$,  for
$1\le i_1 < \cdots <i_k \le n$, $\pi =x_{i_1}^{s_1}\cdots
x_{i_t}^{s_t}$ is called a monomial. The degree of $\pi$, denoted
by $\mbox{deg}(\pi)$, is $\sum\limits^t_{j=1}s_j$. $\pi$ is multilinear,
if $s_1 = \cdots = s_t = 1$, i.e., $\pi$ is linear in all its
variables $x_{i_1}, \dots, x_{i_t}$. For any given integer $q\ge
2$, $\pi$ is called a $q$-monomial if $1\le s_1, \dots, s_t \le q-1$.
In particular, a multilinear monomial is the same as  a $2$-monomial.

An arithmetic circuit, or circuit for short, is a directed acyclic
graph consisting of $+$ gates with unbounded fan-ins, $\times$ gates with two
fan-ins, and terminal nodes that correspond to variables. The size,
denoted by $s(n)$, of a circuit with $n$ variables is the number
of gates in that circuit. A circuit is considered a tree-like circuit
 if the fan-out of every gate is at most one, i.e.,
the underlying directed acyclic graph that excludes all the terminal nodes is a tree.
In other words, in a tree-like circuit, only the terminal nodes can have more than one fan-out
(or out-going edge).

Throughout this paper, the $O^*(\cdot)$ notation is used to
suppress $\mbox{poly}(n,k)$ factors in time complexity bounds.

By definition, any polynomial $F(x_1,\dots,x_n)$ can be expressed
as a sum of a list of monomials, called the sum-product expansion.
The degree of the polynomial is the largest degree of its
monomials in the expansion. With this expanded expression, it is trivial to
see whether $F(x_1,\dots,x_n)$ has a multilinear monomial, or a
monomial with any given pattern. Unfortunately, such an expanded expression is
essentially problematic and infeasible due to the fact that a
polynomial may often have exponentially many monomials in its
sum-product expansion. The challenge then is to test whether $F(x_1,\dots,x_n)$
has a multilinear monomial, or any other desired monomial, efficiently but
without expanding it into its sum-product representation.

For any integer $k \ge 1$, we consider the group
$Z^k_2$ with the multiplication $\cdot$ defined as follows. For
$k$-dimensional column vectors $\vec{x}, \vec{y} \in Z^k_2$ with
$\vec{x} = (x_1, \ldots, x_k)^T$ and $\vec{y} = (y_1, \ldots,
y_k)^T$, $\vec{x} \cdot \vec{y} = (x_1+y_1, \ldots, x_k+y_k)^T.$
$\vec{v}_0=(0, \ldots, 0)^T$ is the zero element in the group.
For any field $\mathcal{ F}$, the group algebra $\mathcal{F}[Z^k_2]$ is defined as
follows. Every element $u \in \mathcal{F}[Z^k_2]$ is a linear sum of the form
\begin{eqnarray}\label{exp-2}
 u &=& \sum_{\vec{x}_i\in Z^k_2,~ a_{i}\in \mathcal{F}} a_{i} \vec{x}_i.
\end{eqnarray}
For any element
$v = \sum\limits_{\vec{x}_i\in Z^k_2,~ b_{i}\in \mathcal{F}} b_{i}
\vec{x}_i$,
We define
\begin{eqnarray}
 u + v  &=& \sum_{a_{i},~ b_{i}\in \mathcal{F},~  \vec{x}_i\in Z^k_2}  (a_i+b_i)
 \vec{x}_i, \  \mbox{and} \nonumber\\
u \cdot v &=& \sum_{a_i,~ b_j\in \mathcal{F},~ \mbox{ and }~\vec{x}_i,~ \vec{y}_j\in Z^k_2}  (a_i b_j)
(\vec{x}_i\cdot \vec{y}_j). \nonumber
\end{eqnarray}
For any scalar $c \in \mathcal{F}$,
\begin{eqnarray}
 c u &=& c \left(\sum_{\vec{x}_i\in Z^k_2, \ a_i\in \mathcal{F}} a_{i} \vec{x}_i\right)
 = \sum_{\vec{x}_i\in Z^k_2,\  a_{i}\in \mathcal{F}} (c a_{i})\vec{x}_i. \nonumber
\end{eqnarray}
The zero element in the group algebra $ \mathcal{ F}[Z^k_2]$ is
$ {\bf  0} = \sum_{\vec{v}} 0\vec{v}$, where $0$ is the zero element in $\mathcal{ F}$
and $\vec{v}$ is any vector in $\displaystyle Z_2^k$. For example,
${\bf  0} = 0\vec{v_0} = 0\vec{v}_1 + 0\vec{v}_2 + 0\vec{v}_3$,
for any $\displaystyle \vec{v}_i \in Z^k_2$, $1\le i\le 3$.
The identity element in the group algebra $\displaystyle \mathcal{ F}[Z^k_2]$ is
$ {\bf 1} = 1 \vec{v}_0 =  \vec{v}_0$, where $1$ is the identity element in $\mathcal{ F}$.
For any vector $\vec{v} =(v_1, \ldots, v_k)^T \in Z_2^k$, for
$i\ge 0$, let
$ (\vec{v})^i = (i v_1, \ldots, i v_k)^T.$
 In particular, when the field $\mathcal{ F}$ is  $Z_2$ (or in general, of characteristic $2$),
 in the group algebra $\mathcal{ F}[Z_2^k]$,
for any $\vec{z}\in Z_2^k$ we have $(\vec{v})^0 = (\vec{v})^2 =
\vec{v}_0$, and $\vec{z} + \vec{z} = \vec{0}$.

\section{A New Transformation}

In this section, we shall design a new method to transform any given polynomial $F$ represented by a circuit $\mathcal{ C }$ to a new polynomial
$G$ represented by a new circuit $\mathcal{ C''}$ such that
the $q$-monomial testing problem for $F$ is reduced to the multilinear monomial testing problem for $G$.
This method is an extension of the circuit reconstruction and randomized variable replacement methods proposed by us in \cite{bill12}.

To simplify presentation, we assume that if any given
polynomial has $q$-monomials in its sum-product expansion, then the degrees
of those multilinear monomials are at least
$k$ and one of them has degree exactly $k$. This assumption is feasible,
because when a polynomial has $q$-monomials of degree $\le k$, e.g., the least degree of those is $\ell$ with $1\le \ell < k$,
then we can multiply
the polynomial by a list of $k-\ell$ new variables so that the resulting polynomial
will have $q$-monomials with degrees satisfying the aforementioned assumption.

\subsection{A New Circuit Reconstruction Method}\label{CR}

In this section and the next, we shall extend the transformation methods designed in \cite{bill12} to general circuits.
The circuit reconstruction and variable replacement methods developed by us in \cite{bill12} work for
tree-like circuits only. In essence, the methods are as follows:
Replace each original variable $x$ in the polynomial by a $+$ gate $g$; for each outgoing edge of $x$, duplicate a copy of
$g$; for each $g$, allow it to receive inputs from $q-1$ many new $y$-variables; for each edge from a $y$-variable to a duplicated
gate $g$, replace it with a new $\times$  gate that receives inputs from the $y$-variable and a new $z$-variable that then feeds the output to
$g$. Additionally, the methods add a new $\times$ gate $f'$ that multiplies the output of $f$ with a new $z$-variable
for each $\times$ gate $f$ in the original circuit.

For any given polynomial $F(x_1,x_2,\ldots,x_n)$ represented by a circuit $\mathcal{ C}$ of size $s(n)$,
we first reconstruct the circuit $\mathcal{ C}$ in three steps as follows:

{\bf Duplicating $+$ gates.} Starting at the bottom layer of the circuit $\mathcal{ C}$, for each $+$ gate $g$ with outgoing edges
$f_1, f_2,\ldots, f_{\ell}$, replace $g$ with  $\ell$ copies  $g_1,g_2,\ldots,g_{\ell}$ such that
each $g_i$ has the same input as $g$, but the only outgoing edge of $g_i$ is $f_i$, $1\le i\le \ell$.

{\bf Duplicating terminal nodes.} For each variable $x_i$,
if $x_i$ is the input to a list of gates $g_1, g_2, \ldots, g_{\ell}$, then create $\ell$
terminal nodes $u_1, u_2, \ldots, u_{\ell}$ such that each of them represents
a copy of the variable $x_i$ and $g_j$ receives input from $u_j$, $1\le j\le \ell$.

Let $\mathcal{ C}^*$ denote the reconstructed circuit after the above two reconstruction steps.
Obviously, both circuits $\mathcal{ C}$ and $\mathcal{ C^*}$ compute the same polynomial $F$.

{\bf Adding  new $\times$ gates and new variables}.
Having completed the reconstruction to obtain $\mathcal{ C}^*$,
we then expand it to a new circuit $\mathcal{ C'}$ as follows.
For every edge $e_i$ in $\mathcal{ C^*}$ (including every edge between a gate and a terminal node)
such that $e_i$ conveys the output of $u_i$ to $v_i$,
add a new $\times$ gate $g_i$ that multiplies the output of $u_i$ with a new variable $z_i$ and passes
the outcome to $v_i$.

Assume that a list of $h$ new $z$-variables
$z_1, z_2, \ldots, z_h$ have been introduced into the circuit $\mathcal{ C'}$.
Let $F'(z_1, z_2, \ldots, z_h, x_1, x_2,
\ldots, x_n)$ be the new polynomial represented by $\mathcal{ C'}$.

\begin{example}\label{example-1}
Consider $F(x_1,x_2) = 16x_1^5 + 32x_1^3x_2 + 2x_1^2x_2 + 16x_1x_2^2 + 2x_2^2$.
Figure 1 shows the circuit $\mathcal{ C}$ that computes $F(x_1,x_2)$.
Figures 2 and 3 show  the circuit $\mathcal{ C^*}$ and the circuit $\mathcal{ C'}$, respectively.
\end{example}

\begin{lemma}\label{rtm-lem1}
Let the $t$ be the length of longest path from the root gate of $\mathcal{ C}$ to its terminal nodes.
$F(x_1,x_2,\ldots,x_n)$ has a monomial $\pi$ of degree $k$  in its sum-product expansion if and only if there is
a monomial $\alpha \pi$ in the sum-product expansion of $F'(z_1, z_2, \ldots, z_h, x_1, x_2,\ldots, x_n)$
such that $\alpha$ is a multilinear monomial of $z$-variables with degree $\le tk + 1$.
Furthermore, if $\pi$ occurs more than once in the sum-product expansion of $F'$, then
every occurrence of $\pi$ in $F'$ has a unique coefficient $\alpha$; and any two different monomials of $x$-variables
in $F'$ will have different coefficients that are multilinear products of $z$-variables.
\end{lemma}

\begin{proof}
Recall that, by the reconstruction processes, $\mathcal{ C^*}$ computes exactly the same polynomial $F$.
If $F$ has a monomial of degree $k$,
then let $\mathcal{ T}$ be the sub-circuit of $\mathcal{ C^*}$ that generates the monomial $\pi$,
and $\mathcal{ T'}$ be the corresponding sub-circuit in $\mathcal{ C'}$. By the way by which the new $z$-variables are introduced,
the monomial generated by $\mathcal{ T'}$ is $\alpha\pi$ with $\alpha$ as the product of all the $z$-variables
added to the edges of $\mathcal{ T}$
to yield $\mathcal{ T'}$. Since $\pi$ has degree $k$, $\mathcal{ T}$ has $k$ terminal nodes, corresponding to $k$ paths from
 the root to those terminal nodes. Thus, $\mathcal{ T}$ has at most $tk$ edges.
Note that one additional $z$-variable is added to the output edge of the root gate. This implies that
$\alpha$ is a multilinear monomial of $z$-variables with degree $\le tk+1$.

\begin{center}
\begin{tabular}{cc}
\epsfxsize=1.5in
\hspace*{.0in}{\epsfbox{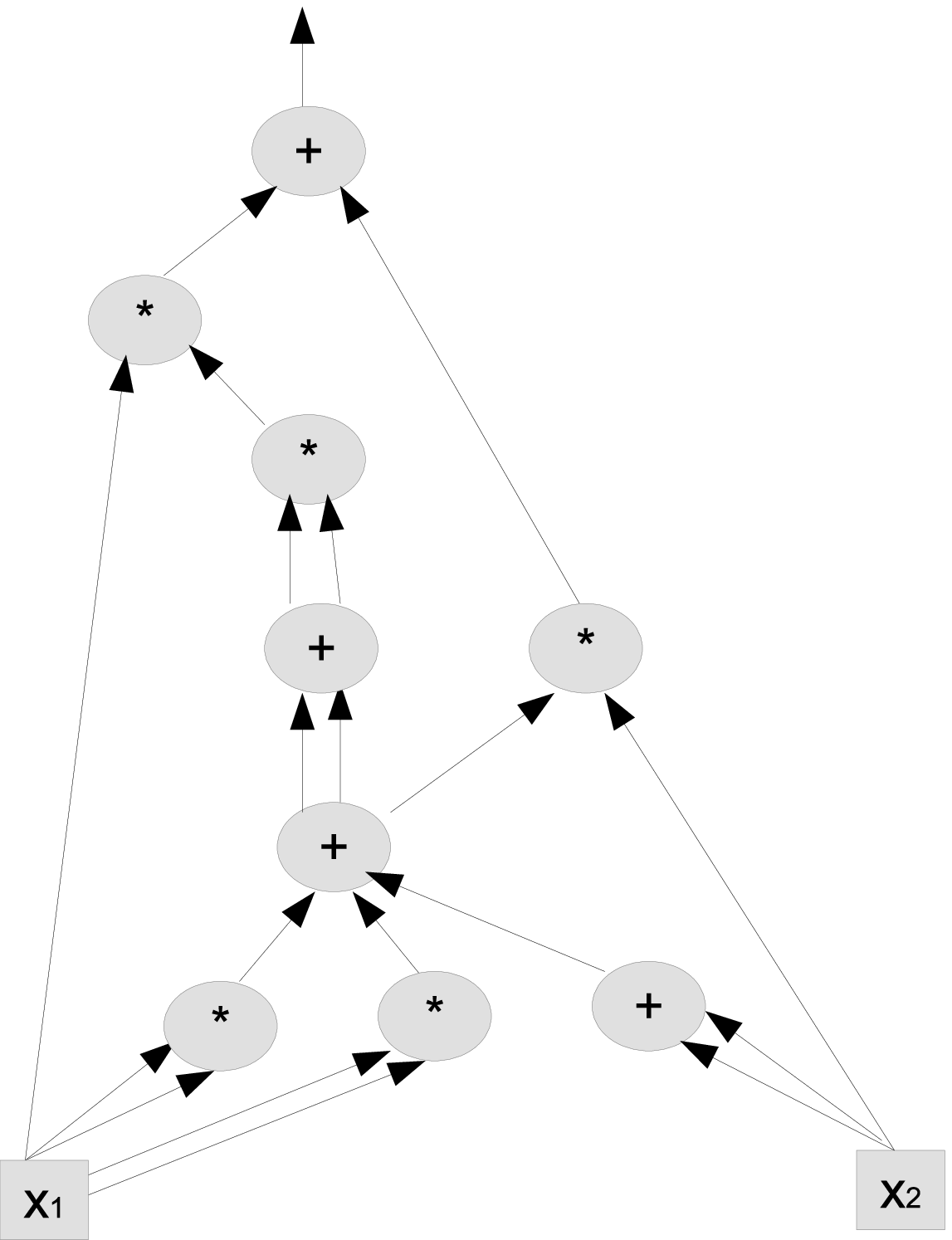}} 
&
\epsfxsize=1.8in
\hspace*{.2in}{\epsfbox{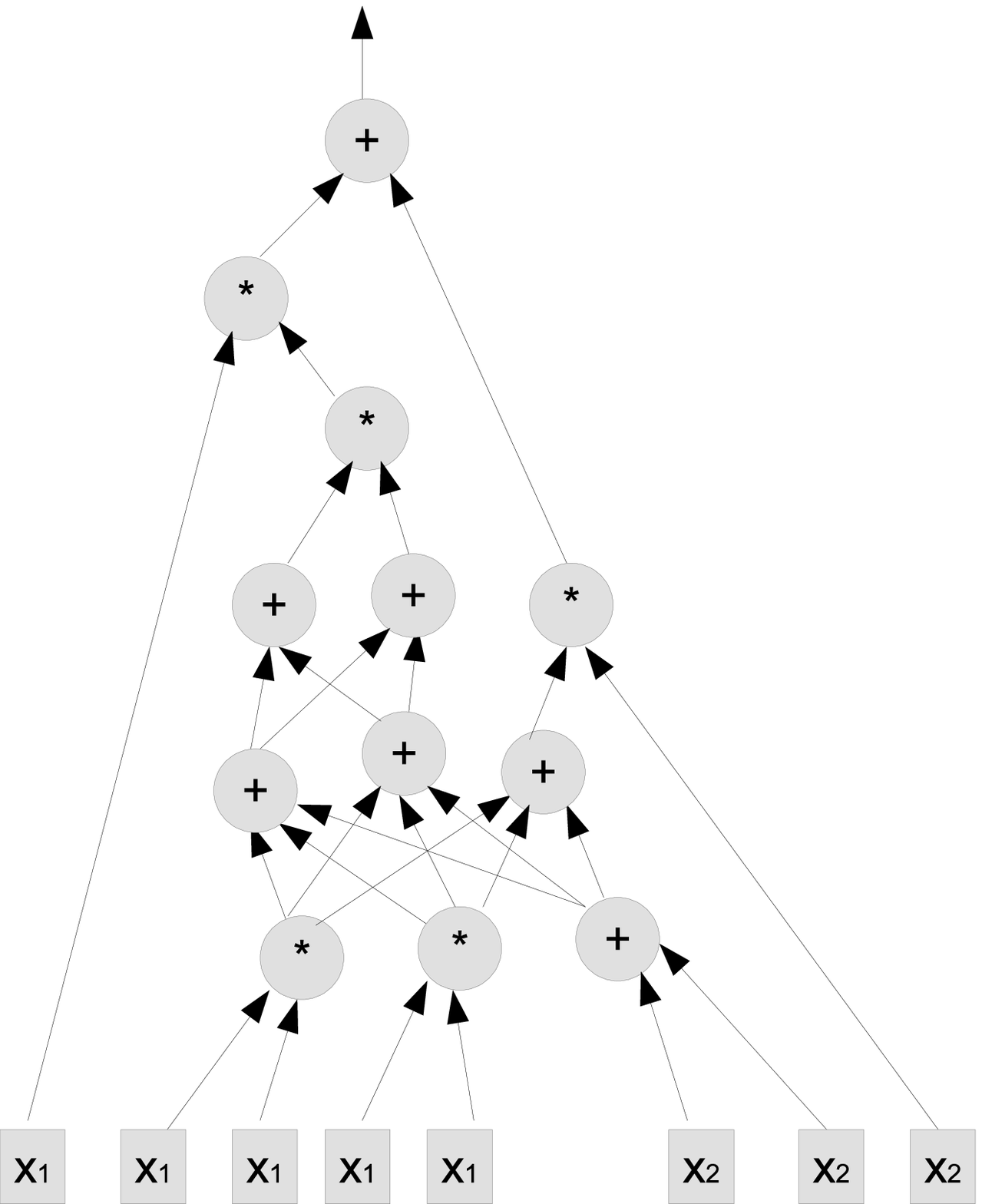}} 
\\
{\bf Fig. 1.} Circuit $\mathcal{ C}$ for $F(x_1,x_2)$
&
\hspace*{.2in}{\bf Fig. 2.} Circuit $\mathcal{ C^*}$ for $F(x_1,x_2)$
\end{tabular}
\end{center}

\begin{center}
\begin{tabular}{c}
\epsfxsize=2.5in
\hspace*{0in}{\epsfbox{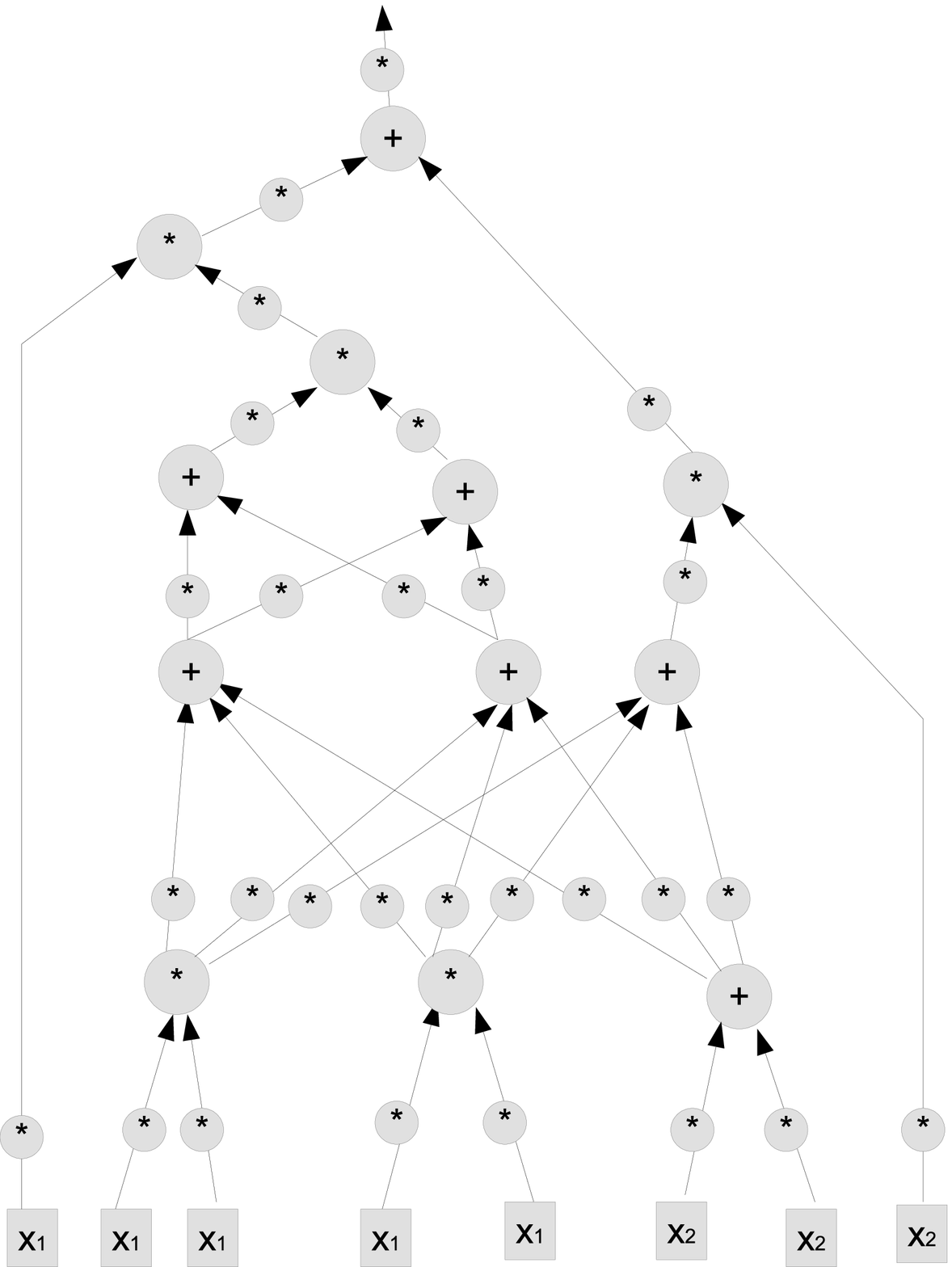}} 
\\
{\bf Fig. 3.} Circuit $\mathcal{ C'}$ for $F(x_1,x_2).$
Due to space limitation, \\
all $z$-variables for the new $\times$ gates
are not shown in the figure.
\end{tabular}
\end{center}

If $F'$ has a monomial $\alpha\pi$ such that $\alpha$ is a product of $z$-variables and $\pi$ is a product of $x$-variables,
then let $\mathcal{ M}'$ be the sub-circuit of $\mathcal{ C'}$ that generates $\alpha\pi$. According to the construction of $\mathcal{ C^*}$
and $\mathcal{ C'}$, removing all the $z$-variables
along with the newly added $\times$ gates from $\mathcal{ M'}$ will result in a sub-circuit $\mathcal{ M}$ of $\mathcal{ C}^*$
that generates $\pi$. Thereby, $\pi$ is a monomial in $F$.

Now, consider that $F'$ has two monomials $\alpha\pi$ and $\beta\phi$ such that, $\pi$ and $\phi$ are products of
$x$-variables and $\alpha$ and $\beta$ are products of $z$-variables. Let $\mathcal{ T}'_1$ and $\mathcal{ T'}_2$
be the sub-circuits in $\mathcal{ C'}$  that generate $\alpha\pi$ and $\beta\phi$, respectively.
Again, according to the construction of $\mathcal{ C^*}$
and $\mathcal{ C'}$, removing all the $z$-variables
along with the newly added $\times$ gates from $\mathcal{ T'}_1$ and $\mathcal{ T'}_2$
will result in two sub-circuits $\mathcal{ T}_1$ and $\mathcal{ T}_2$ of $\mathcal{ C}^*$
that generate $\pi$ and $\phi$, respectively.
When $\pi \not= \phi$, $\mathcal{ T}_1$ and $\mathcal{ T}_2$ are different sub-circuits, this implies that
there is at least an edge $e$ that is in either $\mathcal{ T}_1$ or  $\mathcal{ T}_2$, but not both.
Since a new $\times$ gate is added for $e$ with a new $z$-variable,
there is at least one $z$-variable that is in either $\mathcal{ T'}_1$ or $\mathcal{ T'}_2$, but not both.
Hence, $\alpha$ and $\beta$ do not share the same set of $z$-variables,
because $z$-variables are one to one correspondent to the edges in a sub-circuit.
Hence, $\alpha \not= \beta.$ Also, since the $z$-variables in $\alpha$ correspond to edges in $\mathcal{ T'}_1$,
$\alpha$ is multilinear. Similarly, $\beta$ is also multilinear.

Combining the above analysis completes the proof for the lemma.
\end{proof}

\subsection{Variable Replacements}\label{VR}

Following Subsection \ref{CR}, we continue to address how to further transform
the new polynomial $F'(z_1,z_2,\ldots,z_h,x_1,x_2,\ldots,x_n)$ computed by the circuit
$\mathcal{ C'}$. The method for this part of the transformation is similar to, but different from,
the method proposed by us in \cite{bill12}.

{\bf Variable replacements:}\label{Trans}
Here, we start with the new circuit $\mathcal{ C'}$
that computes $F'(z_1, z_2, \ldots, z_h, x_1, x_2,\ldots, x_n)$.
For each variable $x_i$, we replace it with a "weighted" linear sum  of $q-1$  new $y$-variables
$y_{i1},y_{i2},\ldots,y_{i(q-1)}$.
The replacements work as follows:
For each variable $x_i$, we first
add $q-1$ new terminal nodes that represent  $q-1$ many $y$-variables $y_{i1},y_{i2},\ldots,y_{i(q-1)}$.
Then, for each terminal node $u_j$ representing $x_i$ in $\mathcal{ C'}$, we replace $u_j$ with a $+$ gate.
Later, for each new $+$ gate $g_j$ that is created for $u_j$ of $x_i$, let $g_j$ receive input from $y_{i1},y_{i2},\ldots,y_{i(q-1)}$.
That is, we add an edge from each of such $y$-variables to $g_j$. Finally, for each edge $e_{ij}$ from $y_{ij}$ to $g_j$,
replace $e_{ij}$ by a new $\times$ gate that takes inputs from  $y_{ij}$ and a new $z$-variable $z_{ij}$ and sends the output to $g_j$.

Let $\mathcal{ C''}$ be the circuit resulted from the above transformation, and
$$
G(z_1,\ldots,z_h,y_{11},\ldots,y_{1(q-1)},\ldots,y_{n1},\ldots,y_{n(q-1)})
$$
be the polynomial computed by the circuit $\mathcal{ C''}$.

\begin{example}\label{example-2}
We continue Example \ref{example-1} in Subsection \ref{CR}.
The new circuit $\mathcal{ C''}$ for $F(x_1,x_2)$  is given in Figure 4.
\end{example}

\begin{lemma}\label{rtm-lem2}
Let $F(x_1,x_2,\ldots,x_n)$ be any given polynomial represented by
a circuit $\mathcal{ C}$ and $t$ be the length of the longest path of $\mathcal{ C}$. For any fixed integer $q\ge 2$,
$F$ has a $q$-monomial of $x$-variables with degree $k$, then
$G$ has a unique multilinear monomial $\alpha \pi$ such that $\pi$ is a degree $k$ multilinear monomial of $y$-variables
and $\alpha$ is a multilinear monomial of $z$-variables with degree $\le k(t+1) +1 $. If $F$ has no $q$-monomials, then
$G$ has no multilinear monomials of $y$-variables, i.e., $G$ has no monomials of the format $\beta \phi$ such that
$\beta$ is a monomial of $z$-variables and
 $\phi$ is a multilinear monomial of $y$-variables.
\end{lemma}

\begin{center}
\begin{tabular}{c}
\epsfxsize=3in
\hspace*{0in}{\epsfbox{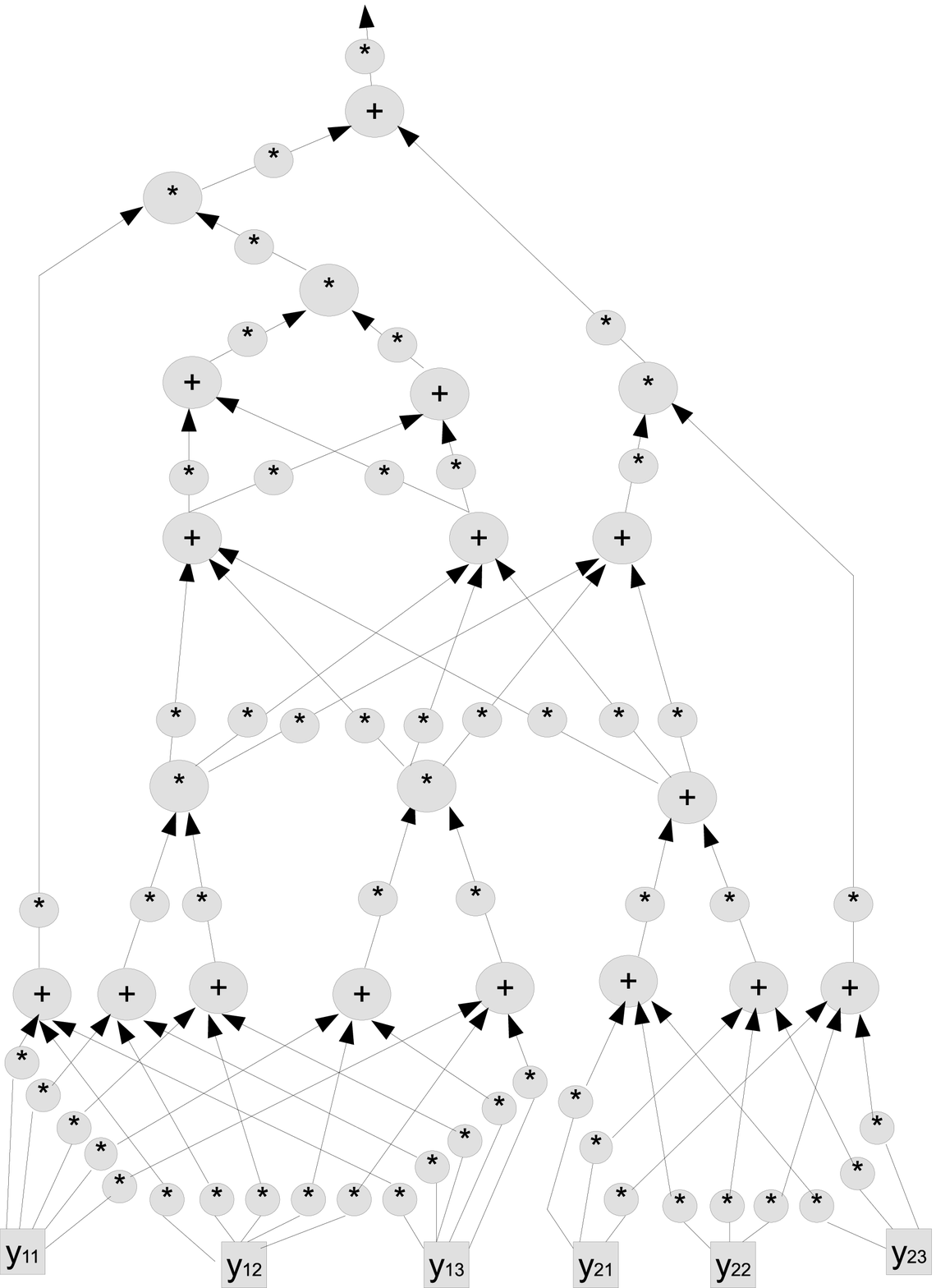}} 
\\
{\bf Fig. 4.} Circuit $\mathcal{ C''}$ for $F(x_1,x_2)$.
Due to space limitation, \\
all $z$-variables for the new $\times$ gates
are not shown in the figure.
\end{tabular}
\end{center}

\begin{proof}
We first show the second part of the lemma, i.e.,
if $F$ has no $q$-monomials, then $G$ has no multilinear monomials of $y$-variables.
Suppose otherwise that $G$ has a multilinear monomial $\phi$ of $y$-variables with a coefficient $\beta$, which is a monomial of
$z$-variables.
Let $\phi = \phi_1\phi_2\cdots \phi_s$ such that $\phi_j$ is the product of all the $y$-variables
in $\phi$ that are used to replace the variable $x_{i_j}$, and let $\mbox{deg}(\phi_j) = d_j$, $1\le j\le  s$.
Consider the sub-circuit $\mathcal{ T}''$ of $\mathcal{ C}''$ that generates $\beta\phi$ when the $x$-variables are replaced by
a "weighted" linear sum of $y$-variables according to the aforementioned variable replacements.
Derive the sub-circuit $\mathcal{ T'}$ in $\mathcal{ C}'$ that corresponds to
$\mathcal{ T}''$ in $\mathcal{ C}''$. Also, derive the sub-circuit $\mathcal{ T}$ in $\mathcal{ C}^*$ that corresponds to
$\mathcal{ T}'$ in $\mathcal{ C}$.
Then, the sub-circuit $\mathcal{ T}$ in $\mathcal{ C}^*$ computes a monomial $\pi = x_{i_1}^{d_i}x_{i_2}^{d_2}\cdots x^{d_s}_{i_s}$ and
$\phi$ is a multilinear monomial in the expansion of the replacement
$$
r(\pi) = \prod^{s}_{j=1 } \left(\prod^{d_j}_{\ell = 1} (z_{j\ell 1}y_{j 1}+z_{j\ell 2}y_{j 2}
+\cdots + z_{j\ell (q-1)}y_{j (q-1)})\right).
$$
which is obtained by the variable replacements described above. If there is one $d_j$ such that
$d_j\ge q$, then let us look at the replacement for $x_{i_j}^{d_j}$, denoted as
$$
r(x_{i_j}^{d_j}) = \prod^{d_j}_{\ell = 1} (z_{j\ell 1}y_{j 1}+z_{j\ell 2}y_{j 2}+\cdots + z_{j\ell (q-1)}y_{j (q-1)}).
$$
Since $d_j\ge q$, by the pigeon hole principle, the expansion of the above $r(x_{i_j}^{d_j})$ has no multilinear monomials.
Thereby, we must have $1\le d_j\le q-1$, $1\le j\le s$. Hence, $\pi$ is a $q$-monomial in $F$, a contradiction to our assumption at
the beginning. Therefore, when $F$ has no $q$-monomials, then $G$ must not have any multilinear monomials of $y$-variables.

We now prove the first part of the lemma.
Suppose $F$ has a $q$-monomial $\pi = x_{i_1}^{s_1}x_{i_2}^{s_2}\cdots x_{i_t}^{s_t}$ with $1\le s_j \le q-1$, $1\le j\le t$.
Let $k = \mbox{deg}(\pi)$.
By Lemma \ref{rtm-lem1}, $F'$ has at least one monomial corresponding to $\pi$. Moreover, each of those monomials
in $F'$ has a format $\alpha \pi$ such that $\alpha$ is a unique product of $z$-variables with
$\mbox{deg}(\alpha) \le tk+1$. Let $\pi' = \alpha \pi$ be one of those monomials.
Consider the sub-circuit $\mathcal{ T'}$ of $\mathcal{ C}'$ that generates $\pi'$. Based on the construction of
$\mathcal{ C}'$, $\mathcal{ T'}$ has $k$ terminal nodes representing $k$ occurrences of all the $x$-variables in $\pi$.
Following the aforementioned variable replacements, each occurrence of those $x$-variables is replaced
by a $+$ gate with inputs from
$q-1$ many $\times$ gates. Moreover, each of such $\times$ gates receives inputs from a $y$-variable and a $z$-variable.
For each $g$ of those $+$ gates, we select one of the $q-1$ many $\times$ gates that are inputs to $g$. Then,
the expanded sub-circuit $\mathcal{ T''}$ of $\mathcal{ T'}$ with all the selected $\times$ gates is a sub-circuit in $\mathcal{ C''}$ that
generates a monomial $\beta\phi$, where $\phi$ is a multilinear monomial of $y$-variables with degree $k$, and
$\beta$ is the product of $\alpha$ with those additional $z$-variables in $\mathcal{ T''}$ but not in
$\mathcal{T}'$, and the degree of
$\beta$ is  $k(t+1) +1$.
\end{proof}

\section{A Faster Randomized Algorithm}

Recently, an $O^*(7.15^k)$ time randomized algorithm has been devised by us in \cite{bill12} for testing $q$-monomials
in any polynomial represented by a tree-like circuit.
We now extend this result to general circuits with a better $O^*(2^k)$ upper bound.

Consider any given polynomial $F(x_1,x_2,\ldots,x_n)$ that is represented by a circuit $\mathcal{ C}$ of size $s(n)$.
Note that the length of the longest path from the root of $\mathcal{ C}$ to any terminal node is no more than $s(n)$.


Let $d = \log_2 (k(s(n)+1)+1) + 1$ and $\mathcal{ F} = \mbox{GF}(2^d)$ be a finite field of $2^d$ many elements.
We consider the group algebra $\mathcal{ F}[Z^k_2]$. Please note that the field $\mathcal{ F} = \mbox{GF}(2^d)$ has characteristic $2$.
This implies that, for any given element $w \in \mathcal{ F}$, adding $w$ for any even number of times yields $0$.
For example, $w + w = 2w = w+w+w+w = 4w = 0.$

The algorithm RTM for testing whether $F(x_1,x_2,\ldots,x_n)$ has a $q$-monomial of degree $k$ is given in the following.

\begin{quote}
{\bf Algorithm RTM} (\underline{R}andomized \underline{T}esting of $q$-\underline{M}onomials):
\begin{description}
\item[1.] As described in Subsections \ref{CR} and \ref{VR},
reconstruct the circuit $\mathcal{ C}$ to obtain $\mathcal{ C}^*$ that computes the same polynomial
$F$ and then introduce new $z$-variables to $\mathcal{ C}^*$
to obtain the new circuit $\mathcal{ C'}$ that computes $F'(z_1,z_2,\ldots,z_h,x_1,x_2,\ldots, x_n)$.
Finally, obtain a circuit $\mathcal{ C''}$ by variable replacements so that $F'$ is transformed to
$$
G(z_1,\ldots,z_h,y_{11},\ldots,y_{1(q-1)},\ldots,y_{n1},\ldots,y_{n(q-1)}).
$$
\item[2.] Select uniform random vectors $\vec{v}_{ij}
\in Z^k_2-\{\vec{v}_0\}$, and replace the variable $y_{ij}$ with
$(\vec{v}_{ij} + \vec{v}_0)$, $1\le i \le n$ and $1\le j\le q-1$.

\item[3.] Use $\mathcal{ C''}$ to calculate
\begin{eqnarray}\label{rtm-exp2}
G' &=& G(z_1,\ldots,z_h,(\vec{v}_{11}+\vec{v}_0),\ldots,(\vec{v}_{1(q-1)}+\vec{v}_0),
\ldots,\nonumber \\
& & ~~~~(\vec{v}_{n1}+\vec{v}_0),\ldots,(\vec{v}_{n(q-1)}+\vec{v}_0)) \nonumber \\
&= &  \sum_{j=1}^{2^k} f_j(z_1,\ldots,z_h) \cdot \vec{v}_j,
\end{eqnarray}
where each $f_j$ is a polynomial of degree $\le k(s(n)+1)+1$ (see Lemma \ref{rtm-lem2}) over the finite
field $\mathcal{ F}=\mbox{GF}(2^d)$, and $\vec{v}_j$ with $1\le j\le 2^k$ are the $2^k$ distinct vectors in
$Z^k_2$.

\item[4.] Perform polynomial identity testing with the
Schwartz-Zippel algorithm \cite{motwani95} for every
$f_{j}$ over $\mathcal{ F}$.
Return {\em "yes"} if one of those polynomials is not
identical to zero. Otherwise,  return {\em "no"}.
\end{description}
\end{quote}

It should be pointed out that the actual implementation of Step 4 would be running the
Schwartz-Zippel algorithm concurrently for all $f_j$, $1\le j\le 2^k$,
utilizing the circuit $\mathcal{ C''}$. If one of those polynomials is not
identical to zero, then the output of $G'$ as computed by circuit $\mathcal{ C''}$ is not zero.

The group algebra technique established by Koutis \cite{koutis08} assures
the following two properties:

\begin{lemma}\label{rtm-lem3}
(\cite{koutis08})~  Replacing all the variables $y_{ij}$
in $G$ with group algebraic elements $\vec{v}_{ij}+\vec{v}_0$ will make all
monomials $\alpha \pi$ in $G'$ to become zero, if $\pi$ is non-multilinear with respect to $y$-variables. Here,
$\alpha$ is a product of $z$-variables.
\end{lemma}

\begin{proof}
Recall that $\mathcal{ F}$ has characteristic $2$.
For any $\vec{v} \in Z^k_2$, in the group algebra $\mathcal{ F}[Z^k_2]$,
\begin{eqnarray}\label{rt-1}
(\vec{v}+\vec{v}_0)^2 &=& \vec{v}\cdot\vec{v} + 2\cdot\vec{v}\cdot\vec{v}_0 + \vec{v}_0\cdot\vec{v}_0 \nonumber \\
 &=&\vec{v}_0+ 2\cdot\vec{v} + \vec{v}_0 \nonumber \\
 &=& 2\cdot \vec{v}_0 + 2\cdot\vec{v} = {\bf 0}.
\end{eqnarray}
Thus, the lemma follows directly from expression (\ref{rt-1}).
\end{proof}

\begin{lemma}\label{rtm-lem4}
(\cite{koutis08})~ Replacing all the variables $y_{ij}$ in $G$
with group algebraic elements $\vec{v}_{ij}+\vec{v}_0$ will make any monomial $\alpha \pi$
to become zero,  if and only if  the vectors $\vec{v}_{ij}$  are linearly dependent in the vector space $Z^k_2$.
Here, $\pi$ is a multilinear monomial of $y$-variables and $\alpha$ is a product of $z$-variables,
Moreover, when $\pi$ becomes non-zero after the replacements,
it will become the sum of all the vectors in the linear space spanned by those vectors.
\end{lemma}

\begin{proof}
The analysis below gives a proof for this lemma.
Suppose $V$  is a set of  linearly dependent vectors
in $Z^k_2$. Then, there exists a nonempty subset $T \subseteq V$ such that $\prod_{\vec{v}\in T} \vec{v}= \vec{v}_0$.
For any $S\subseteq T$, since
$\prod_{\vec{v}\in T} \vec{v} =
(\prod_{\vec{v}\in S} \vec{v}) \cdot (\prod_{\vec{v}\in T-S} \vec{v})$, we have
$\prod_{\vec{v}\in S} \vec{v} = \prod_{\vec{v}\in T-S} \vec{v}$.
Thereby, we have
\begin{eqnarray}
\prod_{\vec{v}\in T }(\vec{v} + \vec{v}_0)
&=& \sum_{S\subseteq T} (\prod_{\vec{v}\in S}\vec{v}) = {\bf 0}, \nonumber
\end{eqnarray}
since every $\prod_{\vec{v}\in S}\vec{v}$ is paired by the same
$\prod_{\vec{v}\in T-S}\vec{v}$ in the sum above and the addition of the pair
is annihilated because $\mathcal{ F} $ has characteristic $2$. Therefore,
\begin{eqnarray}
\prod_{\vec{v}\in V}(\vec{v} + \vec{v}_0)
&=& \left(~\prod_{\vec{v}\in T} (\vec{v}+ \vec{v}_0)\right) \cdot \left(~\prod_{\vec{v}\in V-T}(\vec{v}+ \vec{v}_0)\right)  \nonumber \\
& =&  0 \cdot \left(~\prod_{\vec{v}\in V-T}(\vec{v}+ \vec{v}_0)\right) = {\bf 0}. \nonumber
\end{eqnarray}

Now consider that vectors in $V$ are linearly independent.
For any two distinct subsets $S, T\subseteq V$, we must have
$\prod_{\vec{v}\in T} \vec{v} \not=\prod_{\vec{v}\in S} \vec{v}$, because otherwise
vectors in $S \cup T - (S \cap T)$ are linearly dependent, implying that vectors in $V$ are linearly dependent.
Therefore,
\begin{eqnarray}
\prod_{\vec{v}\in V}(\vec{v} + \vec{v}_0)
&=& \sum_{T\subseteq V}(\prod_{\vec{v}\in T} \vec{v}) \nonumber
\end{eqnarray}
is the sum of all the $2^{|V|}$ distinct vectors spanned by $V$.
\end{proof}

\begin{theorem}\label{thm-rtm}
Let $q>2$ be any fixed integer and $F(x_1,x_2,\ldots,x_n)$ be an $n$-variate polynomial
represented by a circuit $\mathcal{ C}$ of
size $s(n)$. Then, the randomized algorithm RTM can decide whether $F$ has a $q$-monomial
of degree $k$ in its sum-product expansion in time $O^*(2^ks^6(n))$.
\end{theorem}

Since we are often interested in circuits with polynomial sizes in $n$, the time complexity
of algorithm RTM is $O^*(2^k)$ for those circuits.

\begin{proof}
From the introduction of the new $z$-variables to the circuit $\mathcal{ C'}$,
it is easy to see that every monomial in $F'$ has the format $\alpha \pi$, where $\pi$ is a product of
$x$-variables and $\alpha$ is a product of $z$-variables. Since only $x$-variables are replaced by their respective "weighted" linear sums of
new $y$-variables as specified in Subsection \ref{VR}, monomials in $G$ have the format
$\beta \phi$, where $\phi$ is a product of $y$-variables and $\beta$ is a product of $z$-variables.

Suppose that $F$ has no $q$-monomials. By Lemma \ref{rtm-lem2}, $G$ has no monomials $\beta\phi$ such
that $\phi$ is multilinear with respect to $y$-variables. Moreover, by Lemma \ref{rtm-lem3},
replacing $y$-variables by group algebraic elements at Step 2 will make $\phi$ in every monomial $\beta\phi$ in $G$ to become zero.
Hence, the group algebraic replacements will make  $G$ to become zero and so the algorithm RTM will return {\em "no"}.

Assume that $F$ has a $q$-monomial of degree $k$. By Lemma \ref{rtm-lem2},
$G$ has a monomial $\beta \phi$ such that
$\phi$ is a multilinear monomial of degree $k$ with respect to $y$ variables and
$\beta$ is a multilinear monomial of degree $\le k(s(n)+1)+1$ with respect to $z$-variables.
It follows from a lemma in \cite{blum95} (see also, \cite{bill12}) , that
a list of uniform random vectors from $Z^{k}_2$ will
be linearly independent with probability at least $0.28$.
By Lemma \ref{rtm-lem4}, with probability at least $0.28$, the multilinear monomial
$\phi$ will not be annihilated by the group algebraic replacements at Step 2.
Precisely, with probability at least $0.28$, $\beta\phi$ will become
\begin{eqnarray}\label{rtm-exp7}
\lambda(\beta\phi) & = & \sum^{2^k}_{i=1} \beta \vec{v}_i,
\end{eqnarray}
where $\vec{v}_i$ are distinct vectors in $Z^k_2$.

Let $\mathcal{ S}$ be the set of all those multilinear monomials $\beta\phi$ that survive
the group algebraic replacements for $y$-variables in $G$. Then,
\begin{eqnarray}\label{rtm-exp8}
G'&=&G(z_1,\ldots,z_h,(\vec{v}_{11}+\vec{v}_0),\ldots,(\vec{v}_{1(q-1)}+\vec{v}_0),\ldots,\nonumber \\
&&~~~~~~~~~~~(\vec{v}_{n1}+\vec{v}_0),\ldots,(\vec{v}_{n(q-1)}+\vec{v}_0)) \nonumber \\
&= & \sum_{\beta\phi \in \mathcal{ S}} \lambda(\beta\phi)  \nonumber \\
&= &\sum_{\beta\phi\in \mathcal{ S}} \left(\sum^{2^k}_{i=1}\beta \vec{v}_i\right) \nonumber \\
&= &\sum_{j=1}^{2^k} \left(\sum_{\beta\phi\in \mathcal{ S}} \beta \right) \vec{v}_j
\end{eqnarray}
Let
\begin{eqnarray}
f_j(z_1,\ldots,z_h) & = & \sum_{\beta\phi\in \mathcal{ S}} \beta. \nonumber
\end{eqnarray}
By Lemmas \ref{rtm-lem2} and \ref{rtm-lem3}, the degree of $\beta$ is at most $k(s(n)+1)+1$.
Hence, the coefficient polynomial $f_j$ with respect to $\vec{v}_j$ in $G'$ after the group algebraic replacements
has a degree $\le k(s(n)+1)+1$. Also, by Lemma \ref{rtm-lem2},
$\beta$ is unique with respect to every $\phi$ for each monomial $\beta\phi$ in $G$.
Thus, the possibility of a {\rm "zero-sum"} of coefficients from different
surviving monomials is completely avoided during the computation for
$f_j$. Therefore, conditioned on that $\mathcal{ S}$ is not empty,
$G'$ must not be identical to zero, i.e., there exists at least one
$f_j$ that is not identical to zero. At Step 4, we use the
randomized algorithm by Schwartz-Zippel \cite{motwani95}
to test whether $f_j$ is identical to zero. Since the degree of each $f_j$ is
at most $k(s(n)+1)+1$, it is known that
this testing can be done with probability at least
$1 - \frac{\mbox{deg}(f_j)}{|\mathcal{ F}|} \ge \frac{1}{2}$
in time polynomially in $s(n)$ and $\log_2 |\mathcal{F}| = \log_2(k(s(n)+1)+1) + 1 $. Since $\mathcal{ S}$ is
not empty with probability at least $0.28$, the success probability
of testing whether $G$ has a degree $k$ multilinear monomial of $y$-variables is at
least $0.28 \times \frac{1}{2} > \frac{1}{8}$.

Finally, we address the issues of how to calculate $G'$ and the
time needed to do so. Naturally, every element in the group
algebra $\mathcal{ F}[Z^k_2]$ can be represented by a vector in
$Z^{2^k}_2$. Adding two elements in $\mathcal{ F}[Z^k_2]$ is equivalent to
adding the two corresponding vectors in $Z_2^{2^k}$, and the
latter can be done in $O(2^k\log_2|\mathcal{F}|)$ time via component-wise sum.
In addition, multiplying two elements in $\mathcal{F}[Z^k_2]$ is
equivalent to multiplying the two corresponding vectors in
$Z_2^{2^k}$, and the latter can be done in $O(k2^{k+1}\log_2|\mathcal{F}|)$ with
the help of a similar Fast Fourier Transform style algorithm as in
Williams \cite{williams09}. By the circuit reconstruction and variable replacements
in Subsections \ref{CR} and \ref{VR}, the size of the circuit $\mathcal{ C''}$ is at most
$s^3(n)$. Calculating $G'$ by the circuit $\mathcal{ C''}$ consists of $n * s^6(n)$
arithmetic operations of either adding or multiplying two elements
in $\mathcal{ F}[Z^k_2]$ based on the circuit $\mathcal{ C''}$. Hence, the total
time needed is $O(n*s^6(n) k 2^{k+1}\log_2|\mathcal{F}|)$. At Step 4, we run the
Schwartz-Zippel algorithm on $G'$ to
simultaneously test whether there is one $f_j$ such that $f_j$ is
not identical to zero.
Recall that $\log_2|\mathcal{F}| = log_2 (k(s(n)+1)+1) +1$. The total time for the entire
algorithm is $O^*(2^k s^6(n))$.
\end{proof}

\section{A Deterministic Algorithm via Derandomization}

We shall devise a deterministic algorithm for testing
$q$-monomials in a multivariate polynomial represented by a tree-like circuit.
Our approach is to derandomize Steps 2 and 4 in algorithm RTM
respectively with the help of two advanced techniques of
perfect hashing by Chen {\em et al.} \cite{jianer-chen07} (see also Naor {\em
et al.} \cite{naor95}) and noncommunicative multivariate
polynomial identity testing by Raz and Shpilka \cite{raz05}.
Our approach follows the work in \cite{chen11b,chen12b}.
However, we are no longer require $q$ to be a prime
and also obtain a better time bound.

\begin{definition}\label{def-hash}
(See, Chen et al. \cite{jianer-chen07}) Let $n$ and $k$ be two integers such that $1\le k\le n$. Let
$\mathcal{ A} =\{1, 2, \ldots, n\}$ and $\mathcal{ K} = \{1, 2, \ldots,
k\}$. A $k$-coloring of the set $\mathcal{ A}$ is a function from
$\mathcal{ A}$ to $\mathcal{ K}$. A collection $\mathcal{ F}$ of
$k$-colorings of $\mathcal{ A}$ is a $(n,k)$-family of {\em perfect
hashing functions} if for any subset $W$ of $k$ elements in $\mathcal{
A}$, there is a $k$-coloring $h \in \mathcal{ F}$ that is injective
from $W$ to $\mathcal{ K}$, i.e., for any $x, y \in W$, $h(x)$ and
$h(y)$ are distinct elements in $\mathcal{ K}$.
\end{definition}

Like in the design of algorithm RTM, we assume, without loss of generality, that
when a polynomial has  $q$-monomials in its sum-product expansion,
one of the $q$-monomials has exactly a  degree of $k$ and
all the rest of those will have degrees at least $k$.

\begin{theorem}\label{thm-dtm}
Let $q\ge 2$ be fixed integer. Let $F(x_1,x_2,\ldots,x_n)$ be an
$n$-variate polynomial of degree $k$ represented by a tree-like circuit $\mathcal{ C}$
of size $s(n)$. There is a deterministic $O^*(12.8^ks^6(n))$ time
algorithm to test whether $F$ has a $q$-monomial of degree $k$ in
its sum-product expansion.
\end{theorem}

\begin{proof}
Let $d = \log_2 (k(s(n)+1)+1) + 1$ and $\mathcal{ F} = \mbox{GF}(2^d)$ be a finite field of $2^d$ elements.
The deterministic algorithm DTM for testing whether $F$ has a
$q$-monomial of degree $k$ is given as follows.

\begin{quote}
Algorithm \mbox{DTM} (\underline{D}eterministic
\underline{T}esting of $q$-\underline{M}onomials):
\begin{description}
\item[1.] As in the Algorithm RTM, following circuit reconstruction and variable replacements in
 Subsections \ref{CR} and \ref{VR}, reconstruct the circuit $\mathcal{ C}$ to obtain $\mathcal{ C}^*$
 that computes the same polynomial $F$ and then introduce new $z$-variables to $\mathcal{ C}^*$
to obtain the new circuit $\mathcal{ C'}$ that computes $F'(z_1,z_2,\ldots,z_h,x_1,x_2,\ldots, x_n)$.
Finally, perform variable replacements to obtain the circuit
$\mathcal{ C''}$ that transforms $F'$ to
$$
G(z_1,\ldots,z_h,y_{11},\ldots,y_{1(q-1)},\ldots,y_{n1},\ldots,y_{n(q-1)}).
$$

\item[2.] Construct with the algorithm by Chen {\em at el.}
\cite{jianer-chen07}  a $((q-1)n s(n), k)$-family of perfect hashing functions
$\mathcal{ H}$ of size $O(6.4^k\log_2^2 ((q-1)n s(n)))$

\item[3.] Select $k$ linearly
independent vectors $\vec{v}_1,\ldots,\vec{v}_{k} \in Z^{k}_2$. (No
randomization is needed at this step, either.)

\item[4] For each perfect
hashing function $\lambda \in\mathcal{ H}$ do
\begin{description}
\item[4.1.] Let $\gamma(i,j)$
be any given one-to-one mapping from $\{(i,j) | 1\le i\le n \mbox{\ and\ } 1\le j\le q-1\}$
to $\{1,2,\ldots,(q-1)n\}$ to label variables $y_{ij}$.
Replace each variable $y_{ij}$ in $G$ with
$(\vec{v}_{\lambda(\gamma(i,j))} + \vec{v}_0)$, $1\le i \le n$ and $1\le j\le q-1$.

\item[4.2.] Use $\mathcal{ C''}$ to calculate
\begin{eqnarray}\label{exp-thm-dt}
G'&=& G(z_1,\ldots,z_h,(\vec{v}_{\lambda(\gamma(1,1))}+\vec{v}_0),\ldots,(\vec{v}_{\lambda(\gamma(1,(q-1)))}+\vec{v}_0),
 \nonumber \\
&& \hspace{7mm}
\ldots,(\vec{v}_{\lambda(\gamma(n,1))}+\vec{v}_0),\ldots,(\vec{v}_{\lambda(\gamma(n,(q-1)))}+\vec{v}_0)) \nonumber \\
&=&  \sum_{j=1}^{2^k} f_j(z_1,\ldots,z_h) \cdot \vec{v}_j,
\end{eqnarray}
where each $f_j$ is a polynomial of degree $\le k(s(n)+1)+1$ (see, Lemma \ref{rtm-lem2}) over the finite
field $\mathcal{ F}=\mbox{GF}(2^d)$, and $\vec{v}_j$ with $1\le j\le 2^k$ are the $2^k$ distinct vectors in
$Z^k_2$.

\item[4.3.] Perform polynomial identity testing with the Raz and Shpilka
algorithm \cite{raz05} for every $f_j$ over $\mathcal{ F}$. Stop
and return {\em "yes"} if one of them is not identical to zero.
\end{description}

\item[5.] If all perfect hashing functions  $\lambda \in \mathcal{ H}$  have
been tried without returning {\em "yes"}, then stop and output {\em "no"}.
\end{description}
\end{quote}

The correctness of algorithm DTM is guaranteed by the nature of perfect hashing
and the correctness of algorithm RTM. We shall now focus on analyzing the time complexity
of the algorithm.

Note that $q$ is a fixed constant.
By Chen {\em at el.}\cite{jianer-chen07}, Step 2 can be done in
$O(6.4^k n \log^2 ((q-1)n)) = O^*(6.4^k)$
time. Step 3 can be easily done in $O(k^2)$ time.

It follows from Lemma \ref{rtm-lem3} that all those monomials that are not $q$-monomials in
$F$, and hence in $F'$, will be annihilated when variables $y_{ij}$
are replaced by $(\vec{v}_{\lambda(t(i,j))} + \vec{v}_0)$ in $G$ at Step 4.1.

Consider any given $q$-monomial $\pi = x_{i_1}^{s_1}\cdots
x_{i_t}^{s_t}$ of degree $k$ in $F$ with  $1\le s_j \le q-1$ and
$k=\mbox{deg}(\pi)$,
$j=1,\ldots,t$. By Lemma \ref{rtm-lem2}, there are monomials
$\alpha\pi$ in $F'$ such that $\alpha$ is a multilinear monomial of $z$-variables with degree $\le k(s(n)+1)+1$,
and all such monomials are distinct.
 By Lemma \ref{rtm-lem4}, $\pi$ (hence, $\alpha\pi$) will
survive the replacements at Step 4.1. Let $\mathcal{ S}$ be the
set of all the surviving $q$-monomials $\alpha\pi$. Following the same analysis as
in the proof of Theorem \ref{thm-rtm}, we have
\begin{eqnarray}
G'&=& G(z_1,\ldots,z_h,(\vec{v}_{\lambda(\gamma(1,1))}+\vec{v}_0),\ldots,(\vec{v}_{\lambda(\gamma(1,(q-1)))}+\vec{v}_0),
 \nonumber \\
&& \hspace{7mm}
\ldots,(\vec{v}_{\lambda(\gamma(n,1))}+\vec{v}_0),\ldots,(\vec{v}_{\lambda(\gamma(n,(q-1)))}+\vec{v}_0)) \nonumber \\
&= & \sum_{j=1}^{2^k} \left(\sum_{\beta\phi\in \mathcal{ S}} \beta \right) \vec{v}_j \nonumber \\
& =& \sum_{j=1}^{2^k} f_j(z_1,\ldots,z_h) \vec{v}_j \nonumber \\
& \not= & 0 \nonumber
\end{eqnarray}
since $\mathcal{ S}$ is not empty.
Here,
\begin{eqnarray}
f_j(z_1,\ldots,z_h) & = & \sum_{\beta\phi\in \mathcal{ S}} \beta. \nonumber
\end{eqnarray}
This means that, conditioned on that $\mathcal{ S}$ is not empty,
there is at least one $f_j$ that is not identical to zero.
Again, as in the analysis for algorithm RTM,
the time needed for calculating $G'$ is $O^*(2^ks^6(n))$ when the replacements are fixed for
$x$-variables and the subsequent algebraic replacements are given for $y$-variables.

We now consider imposing noncommunicativity on $z$-variables in $\mathcal{ C}''$.
This can be done by imposing an order for $z$-variable inputs to any gates in $\mathcal{C}''$.
Technically, however,  we shall allow values for $z$-variables
to communicate with those for $y$-variables.
Finally, we use the algorithm by Raz and Shpilka \cite{raz05} to test whether
$f_j(z_1,\ldots,z_h)$ is identical to zero of not. This can be
done in time polynomially in $s(n)$ and $n$, since with the imposed order for $z$-variables
$f_j$ is a non-communicative polynomial represented by a tree-like circuit.

Combining the above analysis, the total time of the algorithm
DTM is  $O^*(6.4^k \times 2^k s^6(n)) = O^*(12.8^k s^6(n))$.

When the circuit size $s(n)$ is a polynomial in $n$, the time bound becomes $O^*(12.8^k).$
\end{proof}

\section{Applications}

We list three applications of the $q$-monomial testing to concrete algorithm designs.
Here, we assume $q\ge 2$ is a fixed integer. Notably, algorithm DTM can help us to derive
a deterministic algorithm for solving the $m$-set $k$-packing problem in $O^*(12.8^{mk})$,
which is, to our best knowledge, the best upper bound for deterministic algorithms
to solve this problem.

\subsection{Allowing Overlapping in $m$-Set $k$-Packing}
Let $\mathcal{ S}$ be a collection of sets so that each member in $\mathcal{ S}$ is a subset of an $n$-element set
$X$. Additional, members in $\mathcal{ S}$ have the same size $m\ge 3$. We may like to ask whether there
are $k$ members in $\mathcal{ S}$ such that those members are either pairwise disjoint or at most $q-1$ members may overlap.
This problem with respect to $q$ is a generalized version of the $m$-Set $k$-packing problem.

We can view each element in $X$ as a variable. Thus, a member in $\mathcal{S}$ is a monomial of $m$ variables.
Let
$$
F(\mathcal{ S}, k) = \left(~\sum_{A \in \mathcal{ S}} f(A)\right)^k,
$$
where $f(A)$ denotes the monomial derived from $A$. Then, the above generalized problem $m$-set $k$-packing with respect to $q$
is equivalent to ask whether $F(\mathcal{ S}, k)$ has a $q$-monomial of degree $mk$. Again, algorithm RTM solves this problem in
$O^*(2^{mk})$ time. When $q=2$, the $O^*(2^{mk})$ bound was obtained in \cite{koutis08}.

Since $F(\mathcal{ S}, k)$ can be represented by a tree-like circuit, we can choose $q=2$ and
apply algorithm DTM to test whether $F(\mathcal{ S}, k)$ has multilinear monomial (i.e., $2$-monomial) of degree $mk$.
Therefore, we have a deterministic algorithm to solve the $m$-set $k$-packing problem
in $O^*(12.8^{mk})$ time. Although there are many faster randomized algorithms for solving this problem,
for deterministic algorithms our $O^*(12.8^{mk})$ upper bound significantly improves the best known upper bound
$O^*(\mbox{exp}(O(mk)))$ by Fellow {\em et al.} \cite{fellow08}. The upper bound in \cite{fellow08}
has a large hidden constant in the exponent, e.g., in the case of $r=3$, their upper bound is
$O^*((12.7D)^{3m})$ for some $D \ge 10.4$.

\subsection{Testing Non-Simple $k$-Paths}
Given any undirected graph $G=(V,E)$ with $|V| = n$, we may like to know
whether there is a $k$-path in $G$ such that the path may have loops but any vertex in the path
can appear  at most $q-1$ times.
It is easy to see that this non-simple $k$-path problem with respect to $q$
is a generalized version of the simple $k$-path problem.

For each vertex $v_i \in V$,
define a polynomial $F_{k,i}$ as follows:
\begin{eqnarray}
F_{1,i}  &=&  x_i, \nonumber \\
F_{k+1,i} &=&  x_i   \left(\sum_{(v_i,v_j)\in E} F_{k,j}\right),~~ k>1. \nonumber
\end{eqnarray}
We define a polynomial for $G$ as
\begin{eqnarray}
F(G, k)  &=&  \sum^{n}_{i=1} F_{k,i}. \nonumber
\end{eqnarray}
Obviously, $F(G,k)$ can be represented by an arithmetic circuit.
It is easy to see that the graph $G$ has a non-simple $k$-path with respect to $q$,
if and only if  $F(G, k)$ has a $q$-monomial of degree $k$. Algorithm RTM can solve this problem in
$O^*(2^k)$ time. When $q=2$, the $O^*(2^k)$ bound was obtained in \cite{koutis08,williams09}.

\subsection{A Generalized $P_2$-Packing Problem}

Given any undirected graph $G=(V,E)$ with $|V| = n$ and an integer $k$, we can collect $P_2$'s from $G$, i.e., simple paths of
length $2$ in $G$. The generalized $P_2$-packing problem with respect to $q$ asks whether there is a collection of
$k$ many $P_2$'s such that either all those $P_2$'s are pairwise disjoint, or at most $q-1$ of them may share a common vertex.
The generalized $P_2$-packing problem with respect to $q$ can be easily transformed to a generalized
$3$-Set $k$-Packing problem with respect to $q$. Thereby, an $O^*(2^{3k})$ time randomized solution is given by algorithm RTM.
When $q=2$, the $O^*(2^3k)$ bound was obtained in \cite{feng11}.

\section*{Acknowledgment}

Shenshi is supported by Dr. Bin Fu's NSF CAREER Award, 2009 April 1 to 2014 March 31.

\end{document}